\documentclass[runningheads,a4paper]{llncs}

\usepackage{amssymb}
\usepackage{graphicx}

\usepackage{url}
\urldef{\mailsa}\path{ mathcci@yahoo.com}
\urldef{\mailsb}\path {dvh@vnu.edu.vn}
\urldef{\mailsc}\path|erika.siebert-cole, peter.strasser, lncs}@springer.com|

\begin{document}

\mainmatter  

\title{Towards Approximate Model Checking DC and PDC Specifications}

\titlerunning{Towards Approximate Model Checking DC and PDC Specifications}

%
%
\author{Changil Choe$^1$
\and Dang Van Hung$^2$ \and Song Han$^3$}
\authorrunning{Changil Choe, Dang Van Hung,  and Song Han}

\institute{$^{1,3}$Faculty of Mathematics, Kim Il Sung University, D.P.R.K\\
\mailsa\\
$^2$Faculty of Information Technology, Vietnam National University, Vietnam\\
\mailsb
}

%
%

\toctitle{Lecture Notes in Computer Science}
\tocauthor{Authors' Instructions}
\maketitle

\begin{abstract}
DC has proved to be a promising tool for the specification and verification of functional requirements on the design of hard real-time systems. Many works were devoted to develop effective techniques for checking the models of hard real-time systems against DC specifications. DC model checking theory is still evolving and yet there is no available tools supporting practical verifications due to the high undecidability of calculus and the great complexity of model checking. Present situation of PDC model checking is much worse than the one of DC model checking. In view of the results so far achieved, it is desirable to develop approximate model checking techniques for DC and PDC specifications. This work was motivated to develop approximate techniques checking automata models of hard real-time systems for DC and PDC specifications. Unlike previous works which only deal with decidable formulas, we want to develop approximate techniques covering whole DC and PDC formulas. The first results of our work, namely, approximate techniques checking real-time automata models of systems for LDI and PLDI specifications, are described in this paper. 
\end{abstract}

\section{Introduction}
Functional requirements and dependability requirements are two kinds of top-level requirements on the design of computing systems which include software embedded hard real-time systems. The functional requirements express what a system must be able to do and what it must not do. The dependability requirements express that the probability for undesirable but unavoidable behavior of a system must be below a certain limit. 

Duration Calculus (abbreviated to DC) was introduced in [1]  as a logic for specifying quantitative timing requirements of hard real-time systems and fully analyzed in [2, 3]. DC has strong expressive power specifiable hard real-time requirements of systems, but its formulas are highly undecidable [4]. Linear duration invariants (abbreviated to LDIs), a decidable subclass of DC formulas, is useful to specify constraints on the durations of states in the systems [5]. A major interest of researchers in DC model checking was to develop effective technique checking timed automata against LDIs and many works were devoted to deal with it [6, 7, 8, 9, 10]. But the algorithms developed so far need complicated preprocessing and huge amounts of computation as well as do not support debugging effectively, although they allow complete verifications in theoretical terms. 

Several researchers defined variants of DC and proposed techniques checking timed state sequence models against some decidable fragments of their calculi [11, 12, 13, 14, 15]. But the complexity of model checking any decidable fragment featuring both negation and chop, DC's only moldality, is non-elementary and thus impractical [15]. Even worse, when such decidable fragments are generalized just slightly to cover more interesting durational constraints the resulting fragments become undecidable [15]. DC model checking theory is still not completed satisfactorily to meet the basic standards for practical application.   

Naturally probabilistic extension of DC, which is called PDC, was studied to specify and verify dependability requirements of hard real-time systems [16, 17]. Some researchers tried to develop a technique checking probabilistic timed automata against so-called probabilistic linear duration invariants (abbreviated to PLDIs) of their calculus called PDC in their paper [18]. Their study did not show considerable results from the complexity point of view, as they noted in their paper. 

DC and PDC which deal with good models and specifications of real-time systems  will be more useful in the design of hard real-time systems if the effective model checking techniques would be available. To the best of our knowledge, very few research results showing the applications of DC model checking in practice were reported until now. 

This work was motivated to develop approximate model checking tools for the verification of automata models of real-time systems against DC and PDC specifications. Approximate model checking is achieved by generating a large number of random paths through the model, evaluating each path for given property, and using the resulting information to generate approximately correct result. Approximate model checking gives the possibility of handling the difficult problems faced in DC and PDC model checking,  such as huge amount of computation and weak debugging capability, as well as gives the possibility of applying undecidable formulas in system verifications. We think that approximate model checking can be a better way to use DC and PDC in the verifications of hard real-time systems than normal model checking, because DC and PDC have strong expressive powers, but they are highly undecidable and the cost of model checking is too high.  

In this paper, we describe our first result by concentrating on showing main idea and its advantage through simple but typical cases of DC and PDC model checking. The rest of the paper is organized as follows. In the next section, we present an approximate technique checking real-time automata against LDIs using genetic algorithm. In section 3, we present an approximate technique checking probabilistic real-time automata against PLDIs, which is based on the technique of section 2. In section 4, we explain about future work.

\section{Approximate Technique Checking Real-time Automata for LDIs }

In this section, we present a technique checking real-time automata against LDIs approximately, based on the genetic algorithm. Genetic algorithm is a good approach to search near-optimal solution, when the problem is so complicated that seeking optimal solution is practically impossible. For our purpose, we define satisfaction relation between a real-time automaton and a LDI slightly differently from other papers, but equivalently in essence. Then, we develop a technique checking real-time automata for LDIs approximately and give an example showing the effectiveness of our technique. At the end of the section we give a remark on our technique.

\subsection{Satisfaction relation between a real-time automaton and a LDI}

\begin{definition}
A real-time automaton  $\mathcal{M}$ is a triple  $\mathcal{M}=(S, T, L)$ consisting of  a finite set   $S$ of states,  a transition relation   $T\subseteq S\times I \times S$
, and  a labeling function $L: S\rightarrow 2^{AP}$assigning a set of atomic propositions to each state $s\in S$.
\end{definition}

Here,  $I$ is the set of closed interval  $[a,b]$   or semi-infinite interval $[a,\infty )$  on  $\mathbb{R}^+$. For the convenience, we simply denote these intervals by  $[*,*]$. Every state of a real-time automaton is both initial state and accepting state.   $AP$ is the set of atomic propositions which is differently decided according to the system. A real-time automaton has one clock which is reset by every transition.

\begin{example}
Gas burner is a device to generate a flame to heat up products using a gaseous fuel. If the flame fails to be on with gas valve is opened, gas leaks. Sensor should detect gas leak and close the gas valve within one second. Then gas valve should not be open within 30 seconds to protect accumulation of gas leakage. Gas may leak again without flame being on at any time after valve is open. The left of Fig.~\ref{fig1} shows real-time automaton model of gas burner. Leak and NLeak are used to denote atoms of gas burner. 
\end{example}

\begin{figure}
\centering
\includegraphics[height=3.4cm]{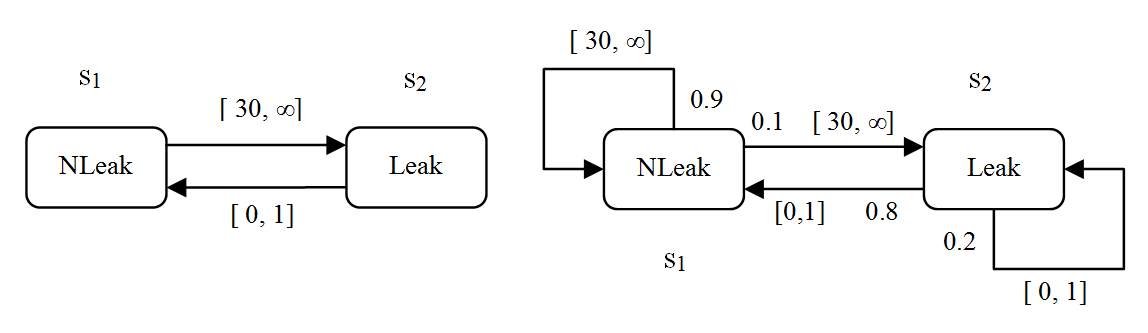}
\caption{Left: Real-time automaton model of gas burner. Right: Probabilistic real-time automaton model of gas burner.}
\label{fig1}
\end{figure}

For a transition  $\rho=(s,[a,b],s')$ of $\mathcal{M}$, notations   $\overleftarrow{\rho}=s$ and   $\overrightarrow{\rho}=s'$ are used.   $ \rho_1\rho_2...\rho_m$ is called a $\emph{sequence}$ and  $ (\rho_1,t_1)(\rho_2,t_2)...(\rho_m,t_m)$ is called a $\emph{time-stamped sequence}$, where $\rho_i=(s_i,[a_i,b_i],s_i^{'})$ and $t_i \in [a_i,b_i]$  for all $i(i\leq i \leq m)$.   $Seq$ and  $TSeq$ are used to denote sequence and time-stamped sequence respectively. 
If a sequence    $ \rho_1\rho_2...\rho_m$ satisfies  $\overrightarrow{\rho_i} = \overleftarrow{\rho}_{i+1}$ for all
$i\, (1\leq i < m)$, it is called a behavior and denoted by $Beh=\rho_1\rho_2...\rho_m$. If a time-stamped sequence  $ (\rho_1,t_1)(\rho_2,t_2)...(\rho_m,t_m)$ satisfies   $\overrightarrow{\rho_i} = \overleftarrow{\rho}_{i+1}$  for all
$i\, (1\leq i < m)$, it is called a time-stamped behavior and denoted by  $TBeh=(\rho_1,t_1)(\rho_2,t_2)...(\rho_m,t_m)$.

\begin{definition}
A DC formula of the form $A\leq \ell \leq B \rightarrow \sum_{i=1}^n c_i \cdot \int P_i \leq C$
is called a linear duration invariant.  
\end{definition}

Here, each  $P_i(1\leq i \leq n)$ is an atomic proposition, $A$ and $B$  are nonnegative real numbers, $B$  could be $\infty$, $c_i (1 \leq i \leq n)$ and $C$  are real numbers \cite{LDI}. LDI says that if the length of observation time interval over $\mathcal{M}$  is between $A$  and $B$, the durations of system sates over that interval should satisfy linear constraint  $ \sum_{i=1}^n k_i \cdot \int P_i \leq C$. Formal semantics of LDI is given in the definition 3.

Let  $TBEH$ be the set of all   $TBeh=(\rho_1,t_1)(\rho_2,t_2)...(\rho_m,t_m)$ of $\mathcal{M}$. Function $L:TBEH \rightarrow \mathbb{R^+}$ is defined as  $L(TBeh) = \sum_{j=1}^m t_j$. For each atomic proposition $P(\in AP)$ of $\mathcal{M}$, function $\int P:TBEH \rightarrow \mathbb{R^+}$  is defined as 
\[\begin{array}{llll}
\int P(TBeh) & = & \sum_{j=1}^m \left
\{\begin{array}{llll} t_j & \,\,\,\,\,\mbox{$P \in \overleftarrow{\rho_j}$}\\
0 & \,\,\,\,\,\mbox{otherwise}
\end{array}\right\}.
\end{array}\]
$\int P(TBeh)$ calculates total duration of $P$-states on $TBeh$, where $P$-state is the state in which $P$  is labeled. For an instance, if    $TBeh=(\rho_1, 3.1)(\rho_2, 2.0)(\rho_3, 1.5)$ and $P$  is labeled on the states $\overleftarrow{\rho_1}$  and $\overleftarrow{\rho_3}$, then $\int P(TBeh)=3.1+1.5=4.6$. Let  $\mathcal{D}$ be a linear duration invariant  over $\mathcal{M}$. Function $LF:TBEH \rightarrow \mathbb{R^+}$ is defined as $LF(TBeh)=\sum_{i=1}^n c_i \cdot \int P_i(TBeh)$. $LF$  is the function calculating the value of linear term  $\sum_{i=1}^n c_i \cdot \int P_i$  of $\mathcal{D}$  for each $TBeh$. Based on these definitions, the satisfaction relation between a real-time automaton $\mathcal{M}$ and a LDI $\mathcal{D}$ is defined as follows.

\begin{definition}
LDI $\mathcal{D}$ is satisfied by real-time automaton $\mathcal{M}$, denoted by $\mathcal{M} \models \mathcal{D}$, iff  $A \leq L(TBeh) \leq B$   implies $LF(TBeh) \leq C$ for all $TBeh$ of $\mathcal{M}$.
\end{definition}

\begin{example}
Fan is installed to protect self-ignition of accumulated gas leakage in gas burner. However, frequent gas leak may cause self-ignition as the ability of fan is limited. A desirable real-time requirement of gas burner is that the proportion of total gas leak time is not more than one twentieth of elapsed time, if the system is observed for more than one minute. This real-time requirement can be specified using LDI as follows.
\[\begin{array}{llll}
\ell \geq 60 \rightarrow 19 \cdot \int Leak - \int NLeak \leq 0
\end{array}\]
Here,  $19 \cdot \int Leak - \int NLeak \leq 0$ is derived from  $\int Leak \leq (1/20) \cdot \ell$ by substituting  $\ell=\int Leak + \int NLeak$.
\end{example}

\subsection{Approximate Technique Checking Real-time Automata for LDIs}

Mathematically, checking a real-time automaton $\mathcal{M}$  for a LDI $\mathcal{D}$  is to solve the following optimization problem.
\[\begin{array}{llll}
$Find out the maximum of LF over $ \{TBeh | A \leq L(TBeh) \leq B \}.
\end{array}\]

If the maximum value of the function $LF$  is smaller than or equal to $C$, then $\mathcal{D}$  is satisfied by real-time automaton $\mathcal{M}$. Unlike previous methods, we check $\mathcal{M}$  approximately using genetic algorithm without any complicated preprocessing and impractical computation. Genetic algorithm works especially well, when the fitness function is linear like LDI. In this subsection, we describe an approximate technique checking real-time automata for LDIs, which is based on the genetic algorithm. We assume that readers have elementary knowledge about technical procedures of genetic algorithm. Given a real-time automaton $\mathcal{M}$  and a LDI  $\mathcal{D}$. \\

$\emph{Encoding.}$

 A time-stamped transition  $(\rho, t)$, where $\rho=(s,[a,b],s')$  and  $t\in [a,b]$ is a $\emph{gene}$ and a time-stamped behavior $(\rho_1,t_1)(\rho_2,t_2)...(\rho_m,t_m)$  is a $\emph{chromosome}$ ($\emph{individual}$).  

$\emph{Fitness\  function.}$

 The linear function $LF$  defined in subsection 2.1 is used as the $\emph{fitness function}$.  $LF$ calculates the value of linear term $\sum_{i=1}^n c_i \cdot \int P_i$  of $\mathcal{D}$  for each individual   $(\rho_1,t_1)(\rho_2,t_2)...(\rho_m,t_m)$. 

$\emph{Initialization.}$ 

The set $BEH$  of all behaviors of $\mathcal{M}$ can be expressed as the union of regular expressions consisting of concatenation and Kleene closure on the alphabet $T$. For example, the set $BEH$ of the behaviors of gas burner can be expressed as $BEH=\rho_1(\rho_2\rho_1)^* \cup \rho_2(\rho_1\rho_2)^* \cup \rho_1(\rho_2\rho_1)^*\rho_2 \cup  \rho_2(\rho_1\rho_2)^*\rho_1$, where  $\rho_1=(s_1, [30, \infty], s_2)$  and $\rho_2=(s_2, [0,1], s_1)$ . Therefore, it's better to choose individuals uniformly from each component of union for quick and uniform expansion of search space, when we create initial population and generate new population. 

$\emph{Selection\ operation.}$ 

Elitist preserving selection which retains the best individuals in a generation unchanged in the next generation is used. 

$\emph{Mutation\ operation.}$ 

Mutation is realized by altering a gene $(\rho, t)$  with another gene $(\rho, t')$  where $\rho=(s,[a,b],s')$  and  $t,t'\in [a,b]$. Multi-point mutation can be used for relatively long individuals. Applications of mutation operation expand the breadth of search space. 

$\emph{Cut\ and\ splice\ operation.}$ 

Cut and splice produces two new individuals from two individuals having same gene, by swapping each suffix beyond the selected gene. \\

Genetic algorithm checking a real-time automaton for a LDI is composed as follows.\\

$\emph{Step 1:}$ 

Using initialization method described above, create initial population $P(0)$ consisting of $N$  individuals and satisfying $A \leq \ell(TBeh) \leq B$ for each individual.

$\emph{Step 2:}$ 

Evaluate fitness of each individual. If $LF(TBeh) > C$  for some individual $TBeh$, terminate the algorithm with output $\mathcal{M} \not\models  \mathcal{D}$. (Note that counter example $TBeh$  is used for debugging.)

$\emph{Step 3:}$ 

Generate population $Q$  by applying genetic operations to current population $P(n)$. Remove all individuals not satisfying    $A \leq \ell(TBeh) \leq B$ from $Q$ and add new individuals satisfying  $A \leq \ell(TBeh) \leq B$  as many as the number of removed individuals. 

$\emph{Step 4:}$  

Generate new population $P(n+1)$  from $Q$  by changing the least-fit individuals of $Q$  with the best-fit individuals of  $P(n)$. 

$\emph{Step 5:}$ 

Repeat step 2-4 until the best-fitness value is settled in the sequence of populations or $n$  is reached to the maximum. 

\subsection{Experiment and Remark}

We applied our genetic algorithm to check the real-time automaton of Example 1 against the LDI of Example 2. Encoding and fitness function were decided according to the above method. Initial population was created by choosing individuals from  $\rho_1(\rho_2\rho_1)^*,   \rho_2(\rho_1\rho_2)^*,  \rho_1(\rho_2\rho_1)^*\rho_2$ and   $\rho_2(\rho_1\rho_2)^*\rho_1$  randomly but uniformly. We executed our genetic algorithm 10 times by changing parameter $N$  between 80-100, $P_m$  between 0.1-0.3 and $P_d$  between 0.4-0.6. Here, $P_m$  is a probability of mutation and $P_d$  is a probability of cut and splice. Termination condition was  $n$=50. 
The best fitness was reached to -3 or nearly -3 in each execution. From this, we could estimate the maximum of $19\cdot \int Leak - \int NLeak $ is -3 which is much smaller than  $C$=0. Consequently, we could confidently conclude that real-time requirement  $ \ell \geq 60 \rightarrow 19\cdot \int Leak - \int NLeak \leq 0$ is satisfied by gas burner model of Example 1.

The approximate technique of this section neither require complicated preprocessing nor need impractical calculation. It also has the advantage of finding out counter examples violating requirement specification, which is achieved by applying algorithm repeatedly. In case that the maximum of $LF$  is different from $C$, it certainly demonstrates same effect with normal model checking. But it is needed more executions of algorithm to get enough information about the maximum of $LF$ in opposite case. Our technique does not largely depend on the increase of state number of system model as the fitness function is linear.

\section{Approximate technique checking probabilistic real-time automata for PLDIs}

To specify the dependability requirements of real-time systems, a kind of probabilistic extension of DC  has been introduced in [16, 17]. No rigorous syntax has been introduced in these papers, and the authors just focused on the development of techniques for reasoning instead of checking. In [18], authors introduced probabilistic duration calculus (abbreviated to PDC) which is a conservative extension of DC and defined its semantics using behavioral of [20]. They also considered the decidability of a class of PDC formulas, so-called probabilistic linear duration invariants (abbreviated to PLDIs), and presented a technique checking probabilistic timed automata against PLDIs. But the chekcing algorithm has too high complexity, as the authors noted in their paper.  

In this section, we present an approximate technique checking probabilistic real-time automata against PLDIs, which uses the technique presented in section 2. For the convenience of approximate model checking, we define the satisfaction relation between a probabilistic real-time automaton and a PLDI differently from [18] but equivalently in essence. 

\subsection{Satisfaction relation between a probabilistic real-time automaton and a PLDI}
A discrete probability distribution over a set $S$ is a mapping $p:S\rightarrow[0,1]$  such that the set $\{s|s\in S, p(s)>0\}$  is finite and $\sum_{s\in S}p(s)=1$. The set of all discrete probability distributions over $S$  is denoted by $S_{Dist}$.

\begin{definition} A probabilistic real-time automaton $\mathcal{M}$  is a triple $\mathcal{M}=(S,D,L)$  consisting of a finite set  $S$ of states,
a probabilistic transition relation  $D:S\rightarrow S_{Dist}\times I$, and a labeling function $L:S\rightarrow 2^{AP}$.
\end{definition}

Every state of a probabilistic real-time automaton is both initial state and accepting state. The discrete probability distribution corresponding to $s$  is denoted by  $p_s$. 

\begin{example}
Realistic gas burner has probabilistic characteristics because senor may fail to detect flame in some cases. From the dependability point of view, gas burner can be modeled as a probabilistic real-time automaton as follows. (See the right of Fig.~\ref{fig1})\\
\[\begin{array}{llll}
S=\{s_1, s_2\}\\
D(s_1)=p_{s_1}\times [30, \infty], p_{s_1}(s_1)=0.9, p_{s_1}(s_2)=0.1\\
D(s_2)=p_{s_2}\times [0,1], p_{s_2}(s_1)=0.8, p_{s_2}(s_2)=0.2\\
L(s_1)=NLeak, L(s_2)=Leak
\end{array}\]
\end{example}

$\rho=(s, p_s(s'), [a,b], s')$  is called a transition and $(\rho, t)$   is called a time-stamped transition, where $t\in [a,b]$.  For example, $\rho=(s_1, 0.1, [30, \infty], s_2)$ is a transition and $(\rho, 31)$ is a time-stamped transition of the probabilistic real-time automaton of Fig. 1. 
$ \rho_1\rho_2...\rho_m...$ is called an infinite behavior and denoted by $Beh$, if $\overrightarrow{\rho_i} = \overleftarrow {\rho_{i+1}}$ for all $i\geq 1$.  $ \rho_1\rho_2...\rho_m$  is called a finite behavior, if  $\overrightarrow{\rho_i} = \overleftarrow {\rho_{i+1}}$  for all  $i=1, ..., m-1$.  $(\rho_1,t_1)(\rho_2,t_2)...(\rho_m,t_m)...$ is called an infinite time-stamped behavior and denoted by $TBeh$, if  $ \rho_1\rho_2...\rho_m...$  is an infinite behavior and $(\rho_i, t_i)$  is a time-stamped transition for each  $i\geq 1$.  $(\rho_1,t_1)(\rho_2,t_2)...(\rho_m,t_m)$ is called a finite time-stamped behavior, if  $ \rho_1\rho_2...\rho_m$  is a finite behavior and $(\rho_i, t_i)$  is a time-stamped transition for each $i=1, ..., m$. 

$BEH(s)$  denotes the set of all infinite behaviors satisfying $\overleftarrow{\rho_1}=s$. $BEH(s)$  can be expressed as a tree structure, denoted by  $G_{BEH(s)}$. For example, The left of Fig.~\ref{fig2} shows the tree expression $G_{BEH(s_1)}$ of $BEH(s_1)$  of the probabilistic real-time automaton described in example 3 (The right of Fig. 1). $G_{BEH(s)}$  can be identified with $BEH(s)$  and we mainly use $G_{BEH(s)}$  for the convenience of description. The infinite behaviors of a probabilistic real-time automaton $\mathcal{M}=(S,D,L)$ can be completely expressed using $|S|$ distinct tree expressions.

\begin{figure}
\centering
\includegraphics[height=5.3cm]{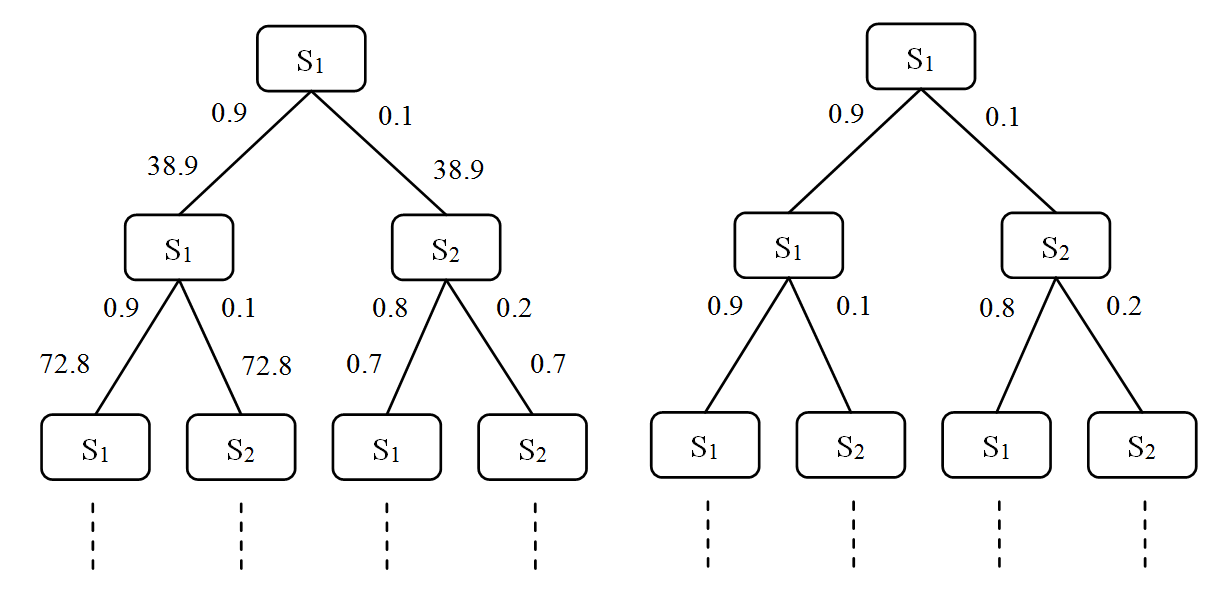}
\caption{Left: Tree expression $G_{TBEH(s_1)}$ of gas burner. Right: Calculation tree $T_{s_1}$ of gas burner.}
\label{fig2}
\end{figure}

Time-stamped instance of $G_{BEH(s)}$, that is, the one which is obtained by changing each time interval of  $G_{BEH(s)}$   with a time point of that interval, is denoted by $G_{TBEH(s)}$  or simply $TBEH(s)$. The probabilistic structure $(G_{TBEH(s)}^{'},$ $ F_{TBEH(s)}, P_{TBEH(s)})$  is defined on $G_{TBEH(s)}$ as follows.  

\begin{itemize}
\item $G_{TBEH(s)}^{'}$ is defined as the set of all infinite paths of  $G_{TBEH(s)}$, starting from the root  $s$.
\item  We denote the set of infinite paths of $G_{TBEH(s)}^{'}$, which have same prefix $(\rho_1,t_1)(\rho_2,t_2)...(\rho_m,t_m)$, denote by $\sigma_{(\rho_1,t_1)(\rho_2,t_2)...(\rho_m,t_m)}$. Here, we consider $s_1\stackrel{p_1, t_1}{\longrightarrow} s_2 \stackrel{p_2,t_2}{\longrightarrow} s_3 \stackrel{p_3,t_3}{\longrightarrow} s_4\ldots $ as $(\rho_1,t_1)(\rho_2,t_2)(\rho_3,t_3)\ldots$ where $\overleftarrow{\rho_i}=s_i$. $F_{TBEH(s)}$ is defined as the smallest $\sigma$-algebra generated by the set of all $\sigma_{(\rho_1,t_1)(\rho_2,t_2)...(\rho_m,t_m)}$.
\item  The probability measure $P_{TBEH(s)}$  on $F_{TBEH(s)}$ is the unique measure such that 
$P_{TBEH(s)}(\sigma_{(\rho_1,t_1)(\rho_2,t_2)...(\rho_m,t_m)})=p_1 \times p_2 \times  \ldots \times p_{m-1}$
\end{itemize}

Note that the probabilistic structure $(G_{TBEH(s)}^{'},$ $ F_{TBEH(s)}, P_{TBEH(s)})$  does not depend on the time stamps. To study temporal and probabilistic behaviors of a probabilistic real-time automaton, we resolve $G_{BEH(s)}$  into time-stamped instances and then study probabilistic behavior of each instance. 

\begin{definition}
 A PDC formula of the form $[D]_{\sqsupseteq \lambda}$  is called a probabilistic linear duration invariant, shortly PLDI, where $\mathcal{D}$  is a linear duration invariant of DC and  $\lambda \in [0,1]$.  
\end{definition}

PLDI means that the possibility that the real-time requirement $\mathcal{D}$  is satisfied by the system is equal to or greater than $\lambda$, even if system runs in the worst case [18]. Formal definition of the semantics of PLDI is as follows.  

\begin{definition}
Let $\mathcal{M}=(S,D,L)$  be a probabilistic real-time automaton and $[\mathcal{D}]_{\sqsupseteq \lambda}$ be a PLDI. 
\begin{itemize}
\item $\mathcal{D}$ is satisfied by $TBeh=(\rho_1,t_1)(\rho_2,t_2)...(\rho_m,t_m)...$, denoted by $TBeh\models \mathcal{D}$, iff $\mathcal{D}$ is satisfied by all finite sub-behaviors $(\rho_i,t_i)(\rho_{i+1},t_{i+1})...(\rho_j,t_j)$ of $TBeh$. Satisfaction relation between $(\rho_i,t_i)(\rho_{i+1},t_{i+1})...(\rho_j,t_j)$ and $\mathcal{D}$ was defined in Section 2.  

\item $[\mathcal{D}]_{\sqsupseteq \lambda}$ is satisfied by $G_{TBEH(s)}$, denoted by $G_{TBEH(s)}\models [\mathcal{D}]_{\sqsupseteq \lambda}$, iff the probability of the set of paths of $G_{TBEH(s)}^{'}$, which satisfy $TBeh\models \mathcal{D}$, is greater than or equal to $\lambda$.

\item $[\mathcal{D}]_{\sqsupseteq \lambda}$ is satisfied by $G_{BEH(s)}$, denoted by $G_{BEH(s)}  \models [\mathcal{D}]_{\sqsupseteq \lambda}$, iff $G_{TBEH(s)}\models [\mathcal{D}]_{\sqsupseteq \lambda}$ for every time-stamped instance  $G_{TBEH(s)}$ of $G_{BEH(s)}$. 

\item  $[\mathcal{D}]_{\sqsupseteq \lambda}$ is satisfied by $\mathcal{M}$, denoted by  $\mathcal{M}  \models [\mathcal{D}]_{\sqsupseteq \lambda}$, iff $G_{BEH(s)}  \models [\mathcal{D}]_{\sqsupseteq \lambda}$ for all $s\in S$.  
\end{itemize}
\end{definition}

\subsection{Approximate Technique Checking Probabilistic real-time automata for PLDIs}
To decide $\mathcal{M}  \models [\mathcal{D}]_{\sqsupseteq \lambda}$  approximately, we introduce the notion of probability calculation tree of a probabilistic rea-time automaton. 

\begin{definition}
Let $\mathcal{M}=(S,D,L)$  be a probabilistic real-time automaton and $s$ be a state of $\mathcal{M}$. The tree constructed according to the following rule is called the probability calculation tree with root $s$ and denoted by $T_s$. 
\begin{itemize}
\item Root is labeled with $s$.

\item Let $v$ be an already constructed vertex with label $s'(\in S)$.

\begin{itemize}
\item For each $s''$ satisfying $p_{s^{'}}(s'')>0$, add new vertex with label $s''$ as a child of $v$. 
\item Label $p_{s'}(s'')$ on the edge connecting the vertex with label $s'$  and the vertex with label $s''$.

\end{itemize}
\end{itemize}
\end{definition}

$|S|$  probability calculation trees are defined for a probabilistic real-time automaton  $\mathcal{M}=(S,D,L)$. Two subtrees having same label are isomorphic each other in a probability calculation tree. And each subtree of $T_s$, whose root is labeled with  $s'$, is isomorphic to the probability calculation tree  $T_{s'}$. 

The right of Fig. 2 shows a probability calculation tree defined from the probabilistic real-time automaton of example 3. As we can see in the figure, the probability calculation tree  $T_{s}$  is obtained by removing time stamps from a $G_{TBEH(s)}$ . That is, the same probability calculation tree  $T_{s}$ is generated from every $G_{TBEH(s)}$. Probabilistic structure is defined on $T_{s}$  in the same way with subsection 3.1. $T_{s}$  is used to check  $G_{BEH(s)}  \models [D]_{\sqsupseteq \lambda}$.

For the simplicity of description, we identify each vertex of a probability calculation tree with its label. Given a finite set $W$ of finite paths of  $T_s$. For a vertex   $v$ of $T_s$, $P_W(v)$  denotes the probability of the set of all infinite paths which start from $v$  and do not include any path in  $W$. 

\begin{theorem}
Given a probability calculation tree $T_s$  and finite set $W$  of finite paths of $T_s$. For each vertex $v$  of $T_s$, $P_W(v)$  is computable.
\end{theorem}

\begin{proof}
Let $w_1, w_2, ..., w_m$  be the paths of $W$, which start from $v$. And let $V={v_1, v_2, ..., v_k}$  be the set which is obtained by eliminating vertices of  $w_1, w_2, ..., w_m$  from the set of all children of all non-end vertices of  $w_1, w_2, ..., w_m$. Then, the following equation holds.
\[\begin{array}{llll}
P_W(v)=p(v,v_1)\times P_W(v_1) + p(v,v_2)\times P_W(v_2) + ... + p(v,v_k)\times P_W(v_k) 
\end{array}\]
Here $p(v,v_i)$  is the multiplication of probability values labeled on the path from $v$  to $v_i$. In case that there is no path starting from $v$  in  $W$, the following equation holds, where $v_1, v_2, ..., v_l$ are children of $v$. 
\[\begin{array}{llll}
P_W(v)=p_v(v_1)\times P_W(v_1) + p_v(v_2)\times P_W(v_2) + ... + p_v(v_k)\times P_W(v_k) 
\end{array}\]
There are only finite different vertices in $T_s$ and we can build linear equation system consisting of such equations described above. Solving it, we can find the value of $P_W(v)$.
\end{proof}

Using Theorem 1, it is possible to decide approximately whether $[\mathcal{D}]_{\sqsupseteq \lambda}$  is satisfied by $\mathcal{M}$  or not. In the rest of this subsection, we describe about it.  
We get real-time automaton $\mathcal{M'}$  from probabilistic real-time automaton $\mathcal{M}$  by removing transition probability values on all edges. We check $\mathcal{M'} \models \mathcal{D}$ using approximate model checking technique of section 2.2. If repeated checking does not detect any finite time-stamped behaviors violating  $\mathcal{D}$   in  $\mathcal{M'}$, we can conclude $\mathcal{M} \models [\mathcal{D}]_{\sqsupseteq \lambda}$. 

Let us now assume that repeated approximate model checking have detected some time-stamped behaviors violating $\mathcal{D}$ in $\mathcal{M'}$ . We get finite time-stamped behaviors of $\mathcal{M}$  by labeling probability values again to all detected time-stamped behaviors of $\mathcal{M'}$. We denote this set by  $W_0$. 

For each $G_{TBEH(s)}$, probability value of the set of behaviors, which start from $s$  and do not include finite time-stamped behaviors of   $W_0$ as part, is different. What we are interested is the minimum of these probability values. If the minimum is equal to or greater than  $\lambda$, we can conclude $\mathcal{M} \models [\mathcal{D}]_{\sqsupseteq \lambda}$, that is, the possibility that $\mathcal{D}$  is satisfied by $\mathcal{M}$ is approximately greater than or equal to $\lambda$  even in the worst case. 

We generate $W_0^{'}$  from $W_0$  by expressing each finite time-stamped behavior of $W_0$  in the form of path and removing time stamps. Each path of $W_0^{'}$  becomes a finite path of a probability calculation tree of  $\mathcal{M}$. The number of elements of  $W_0^{'}$  is smaller than the one of $W_0$  because different time stamps can have same transition probability value. 

Finally, we generate $W=\{ w_1, w_2, ..., w_m \}$  from  $W_0^{'}$  by eliminating each path which include another path. This is because for each calculation tree  $T_s$, the set of infinite paths including $w$  as a subpath is a subset of the set of infinite paths including $w'$ as a subpath, if $w$  includes $w'$  as a subpath. (Note that it is possible to reduce $W$  again for the calculation of $P_W(s)$ in some cases. We don't consider about it in this paper and show an example in the next subsection.)

By applying Theorem 1, we calculate $P_W(s_1), P_W(s_2), ..., P_W(s_n)$   for each state  $s_1, s_2, ..., s_n$ of $\mathcal{M}$. If $P_W(s_i) \geq \lambda$  for all $i(1\leq i \leq n)$, we can conclude that  $\mathcal{M} \models [\mathcal{D}]_{\sqsupseteq \lambda}$  holds approximately. But if $P_W(s_i) < \lambda$   for some  $i(1\leq i \leq n)$, it means that   $\mathcal{M} \not\models [D]_{\sqsupseteq \lambda}$. The technique presented in this subsection can be fully automated. 

\subsection{Experiment and Remark}
Using the technique described above, we decided the satisfaction relation between the probabilistic real-time automaton $\mathcal{M}$  of Example 3 and the probabilistic linear duration invariant  $[\mathcal{D}]_{\sqsupseteq \lambda}$, where $\mathcal{D}$  is  $\ell \geq 60 \rightarrow 19 \cdot \int Leak - \int NLeak \leq 0$  and $\lambda$  is a real number satisfying  $0\leq \lambda \leq 1$. 

For the convenience of consideration, we bounded checking to the time-stamped behaviors whose lengths are not bigger than 8. As a result of 5 repeated application of approximate DC model checking to $\mathcal{M'}$  for $\mathcal{D}$, hundreds of time-stamped behaviors violating  $\mathcal{D}$ were detected. We constructed $W_0$  and $W_0^{'}$  according to the method described above. The number of paths of  $W_0^{'}$ was about 70. 

Finally, we generated $W$  from $W_0^{'}$ according to the method described above, which consists of 4 paths. They are
\[\begin{array}{llll}
s_2\stackrel{0.2}{\longrightarrow}s_2\stackrel{0.2}{\longrightarrow}s_2\stackrel{0.8}{\longrightarrow}s_1\stackrel{0.9}{\longrightarrow}s_1\\
s_2\stackrel{0.2}{\longrightarrow}s_2\stackrel{0.2}{\longrightarrow}s_2\stackrel{0.8}{\longrightarrow}s_1\stackrel{0.1}{\longrightarrow}s_2\\
s_1\stackrel{0.1}{\longrightarrow}s_2\stackrel{0.2}{\longrightarrow}s_2\stackrel{0.2}{\longrightarrow}s_2\stackrel{0.8}{\longrightarrow}s_1\\
s_1\stackrel{0.1}{\longrightarrow}s_2\stackrel{0.2}{\longrightarrow}s_2\stackrel{0.2}{\longrightarrow}s_2\stackrel{0.2}{\longrightarrow}s_2
\end{array}\]

\begin{figure}
\centering
\includegraphics[height=4.6cm]{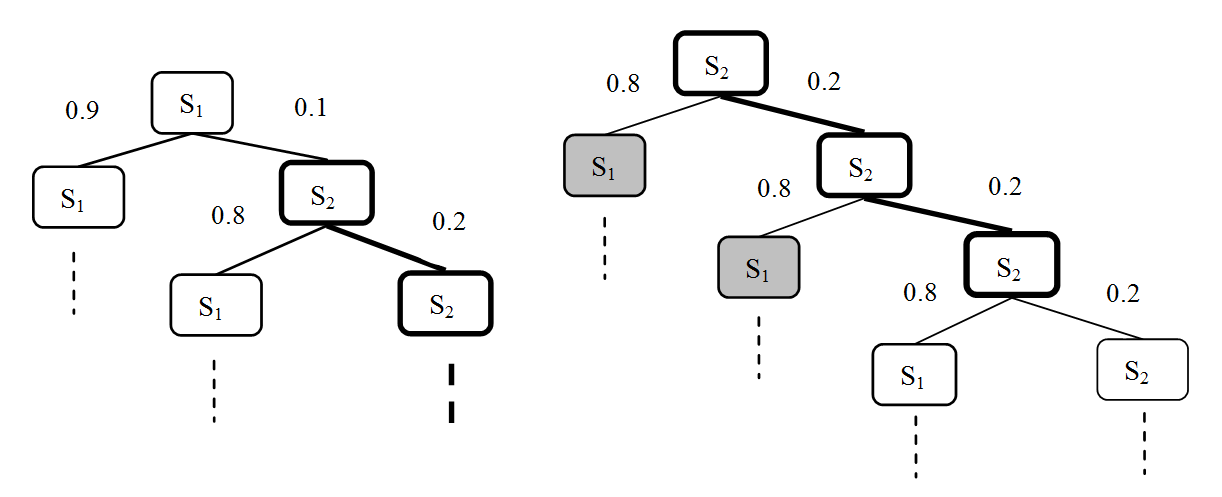}
\caption{Left: Calculation of $P_{W'}(s_1)$. Right: Calculation of  $P_{W'}(s_2)$.}
\label{fig3}
\end{figure}

We didn't apply the Theorem 1 to $W$  in this stage and reduced $W$ again manually in the following way. 4 paths include $s_2\stackrel{0.2}{\longrightarrow}s_2\stackrel{0.2}{\longrightarrow}s_2$  as a subpath.  Therefore, in each probability calculation tree of $\mathcal{M}$, the set of infinite paths including a path of $W$  as a subpath is a subset of the set of infinite paths including  $s_2\stackrel{0.2}{\longrightarrow}s_2\stackrel{0.2}{\longrightarrow}s_2$ as a subpath . 

On the other hand, every infinite behavior of a  $G_{BEH}$, which passes  $s_2$  three consecutive times, has a time-stamped instance violating  $\mathcal{D}$. These mean that we can find an approximate minimum value of the possibility for the satisfaction of $\mathcal{D}$  by $\mathcal{M}$  corresponding to the worst case, by calculating $P_{W'}(s_1)$  and $P_{W'}(s_2)$   where  $W'=\{s_2\stackrel{0.2}{\longrightarrow}s_2\stackrel{0.2}{\longrightarrow}s_2\}$. We applied theorem 1 to this  $W'$. 

As we can see in the left side of Fig.~\ref{fig3},  $s_2\stackrel{0.2}{\longrightarrow}s_2\stackrel{0.2}{\longrightarrow}s_2$  does not start from  $s_1$ of $T_{s_1}$. Therefore, the following equation holds.
\begin{equation}
P_{W'}(s_1) = 0.9 \cdot P_{W'}(s_1) + 0.1 \cdot P_{W'}(s_2)
\end{equation}
In $T_{s_2}$, the right side of Fig.~\ref{fig3},  $V$ consists of two grey-colored  $s_1$-vertices which are children of first two bold lined  $s_2$-vertices. Therefore, the following equation holds.
\begin{equation}
P_{W'}(s_2) = 0.8 \cdot P_{W'}(s_1) + 0.2\cdot 0.8 \cdot P_{W'}(s_1)
\end{equation}
 By combining (1) and (2), we set up the following linear equation system. 
\[\begin{array}{llll}
\left
\{\begin{array}{llll}
P_{W'}(s_1) = 0.9 \cdot P_{W'}(s_1) + 0.1 \cdot P_{W'}(s_2)\\
P_{W'}(s_2) = 0.96 \cdot P_{W'}(s_1)
\end{array}\right.
\end{array}\]
Solving this linear equation system, we have known  $P_{W'}(s_1)=P_{W'}(s_2)=0$. This means that $\mathcal{M}\not\models [\mathcal{D}]_{\sqsupseteq \lambda}$  for any  $\lambda (\in (0,1])$. In other words, the dependability of gas burner for the real-time requirement $\mathcal{D}$  is zero in worst case. 

We tried to make $W$  small as possible, because it can reduce total calculation time considerably. In general, it is needed careful analysis about the system model and requirement specification to minimize $W$.  It can be skipped if $W$  is small. The linear equation system  was homogeneous in the above example. However, it is not homogeneous in general. 

\section{Future Work}
There are no big technical difficulties in adjusting the techniques described in the paper to the timed automata [19] and probabilistic timed automata. For the next stage,  we want to develop approximate technique checking timed automata against undecidable DC formulas containing chop operator which is not considered in normal DC model checking. Discrete measurement operator $\Sigma$ (which is sometimes denoted by $\sharp$) of WDC [21] will be used to represent chop formulas quantitatively to be more convenient for checking. For example, a design requirement $\lceil Leak \rceil ^ \smallfrown \lceil NLeak \rceil ^ \smallfrown \lceil Leak \rceil \rightarrow \ell \geq 30$  of gas burner can be represented as $\Sigma Leak =2 \land \Sigma NLeak=1 \rightarrow \ell \geq 30$. The latter is much more convenient to apply optimization method.

\end{document}